\documentclass[10pt]{svmult}

\usepackage{graphicx}
\usepackage[table,pdftex,hyperref,svgnames]{xcolor}
\usepackage{url,doi}
\usepackage{mathptmx}       
\usepackage{helvet}         
\usepackage{courier}        
\usepackage{type1cm}        
\usepackage[margin = 1.0in]{geometry}
\usepackage{mathtools}
\usepackage{amsmath}
\usepackage{amsfonts}
\usepackage{subfigure}
\usepackage{bm}
\usepackage{amssymb}
\usepackage{epstopdf}
\usepackage{mathtools}  
\usepackage{tabulary}
\usepackage{booktabs}
\usepackage{cite}
\usepackage{enumitem}

\usepackage[utf8]{inputenc}

\usepackage{enumitem}

\DeclareSymbolFont{bbold}{U}{bbold}{m}{n}
\DeclareSymbolFontAlphabet{\mathbbold}{bbold}
\newcommand{\vect}[1]{\mathbbold{#1}}
\newcommand{\vectorones}[1][]{\vect{1}_{#1}}

\newcommand{\argmin}{\operatorname{argmin}}

\newcommand{\Med}{\textup{Med}}

\graphicspath{{./fig/}}

\begin{document}

\setlength{\parskip}{0.3cm}
\setlength\parindent{0em}
\setlist{nolistsep}

\chapter*{Rethinking the Micro-Foundation of Opinion Dynamics: \\Rich Consequences of the Weighted-Median Mechanism}

\vspace{-2cm}{Wenjun Mei*, Francesco Bullo, Ge Chen, Julien M. Hendrickx, Florian D\"{o}rfler}
\vspace{4cm}

\noindent \textbf{Abstract}
\smallskip

To identify the main mechanisms underlying complex opinion formation processes in social systems, researchers have long been exploring simple mechanistic mathematical models. Most existing opinion dynamics models are built on a common micro-foundation, i.e., the weighted-averaging opinion update. However, we argue that this universally-adopted mechanism features a non-negligible unrealistic feature, which brings unnecessary difficulties in seeking a proper balance between model complexity and predictive power. In this paper, we propose the weighted-median mechanism as a new micro-foundation of opinion dynamics, which, with minimal assumptions, fundamentally resolves the inherent unrealistic feature of the weighted-averaging mechanism. Derived from the cognitive dissonance theory in psychology, the weighted-median mechanism is supported by online experiment data and broadens the applicability of opinion dynamics models to multiple-choice issues with ordered discrete options. Moreover, the weighted-median mechanism, despite being the simplest in form, captures various non-trivial real-world features of opinion evolution, while some widely-studied averaging-based models fail to.
\vspace{3cm}

\newpage
\noindent \textbf{MAIN TEXT}
\vspace{1cm}

\section{Introduction}

The key discourse in democratic society starts from exchanges of opinions in deliberative groups, over public debates, or via social media, to eventually reaching consensus or disagreements. Mathematical models help provide mechanistic understandings of how empirically observed macroscopic opinion-formation phenomena emerge from certain microscopic social-influence mechanisms and certain social network structures. Due to the complexity of social interactions, the key challenge in building predictive and mathematically tractable models is to identify the ``salient features'', i.e., the micro-foundations,  that govern the interpersonal influence processes. 

Most existing opinion dynamics models are based upon a common micro-foundation: the weighted-averaging mechanism, also known as the classic \emph{DeGroot model}~\cite{JRPF:56,MHDG:74}. Consider $n$ individuals discussing some issue and each individual $i$'s opinion at time $t$ is denoted by $x_i(t)$. The mathematical form of the DeGroot model is written as
\begin{equation}\label{eq:DeGroot}
  x_i(t+1) = \text{Mean}_i\big(x(t);W\big)=\sum_{j=1}^n w_{ij} x_j(t).
\end{equation}
Here $w_{ij}$ characterizes the influence individual $j$ has on $i$, and the \emph{influence matrix} $W\!=(w_{ij})_{n\!\times\! n}$ induces a directed and weighted graph, referred to as the \emph{influence network} and denoted by $G(W)$, see an example in Fig.~\ref{fig:averaging-WM}(a). Namely, each individual is a node in $G(W)$, and each entry $w_{ij}$ corresponds to a link from $i$ to $j$ with weight $w_{ij}$. By definition,  $w_{ij}\ge 0$ for any $i,\, j$, and $w_{i1}+\dots+w_{in}=1$ for any $i$. The DeGroot model is deceivingly elegant but leads to an overly-simplified prediction that the individuals' opinions reach consensus, i.e., $x_i(t)-x_j(t)\!\to\! 0$ as $t\!\to\! \infty$ for any $i,\,j$, whenever $G(W)$ has a globally reachable and aperiodic strongly connected component~\cite{MHDG:74,AVP-RT:17} (see the Supplementary Section I for a brief review of graph theory). This is a bold conclusion based on mild connectivity conditions. 

The intuition behind DeGroot model's always-consensus behavior is that the weighted-averaging mechanism leads to a non-negligibly unrealistic implication, illustrated via the following simple example and visualized in Fig.~\ref{fig:averaging-WM}(b): Suppose an individual $i$ is influenced by individuals $j$ and $k$ via the weighted-averaging mechanism:
\begin{equation*}
x_i(t+1) = x_i(t) + w_{ik}\big(x_k(t)-x_i(t)\big) + w_{ij}\big(x_j(t)-x_i(t)\big).
\end{equation*}
The equation above implies that whether $i$'s opinion moves towards $x_k(t)$ or $x_j(t)$ is determined by whether $w_{ik}|x_k(t)-x_i(t)|$ is larger than $w_{ij}|x_j(t)-x_i(t)|$. That is, the ``attractiveness'' of opinion $x_j(t)$ to individual $i$ is proportional to the opinion distance $|x_j(t)-x_i(t)|$. Such proportionality implies overly large ``attractive forces'' between distant opinions, which drive the DeGroot model to consensus under mild conditions. Moreover, the notion of opinion distance depends on numerical representation of opinions, which could be arbitrary if the opinions are not numerical by nature.

The weighted-averaging mechanism is widely adopted for its simplicity. However, as indicated by the argument above, such simplicity comes at a cost: In order to explain anything other than consensus, the inherent unrealistic features of weighted averaging have to be first remedied by introducing additional assumptions and parameters that help resist the overly large attractions between distant opinions, see various important extensions of the DeGroot model~\cite{RPA:64,DA-GC-FF-AO:10,NEF-ECJ:90,GD-DN-FA-GW:00,RH-UK:02,PD-AG-DTL:13,TKN-MM-JL:16,NEF-AVP-RT-SEP:16,AVP-RT:17,JL:17,AVP-RT:18,GS-CA-JSB:19}. For instance, Abelson~\cite{RPA:64} assumes that the weights decay with opinion distances. In a more recent paper~\cite{VA-FB-AKS:16b}, individuals with more extreme opinions are assumed to assign more weights to themselves. These modified averaging mechanisms, however, still lead to opinion consensus under mild network connectivity conditions. The Friedkin-Johnsen model~\cite{NEF-ECJ:90} generates disagreement by introducing individuals' persistent attachments to their initial conditions, which resist the attractions by others' opinions. In the biased-assimilation model~\cite{PD-AG-DTL:13}, individuals process weighted averages of others' opinions in a highly nonlinear manner, by weighing confirming evidence more heavily than dis-confirming evidence, which leads to opinion polarization. The bounded-confidence models~\cite{GD-DN-FA-GW:00,RH-UK:02} assume that opinion attractiveness first increases proportionally with opinion distance and is then truncated to zero once the distance exceeds a pre-assumed threshold, which leads to opinion clustering. All the aforementioned models involve additional crucial parameters that need to be identified.

\begin{figure*}[htb]
\centering
\includegraphics[width=0.85\linewidth]{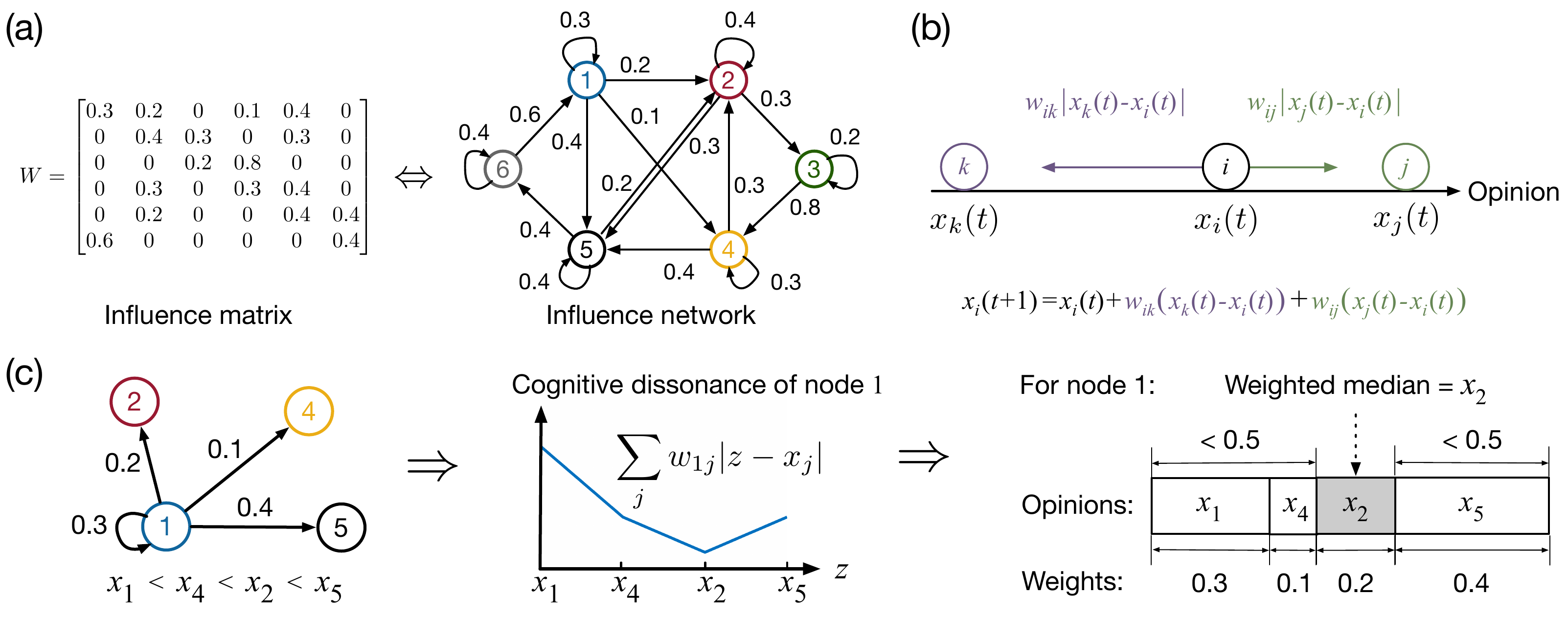}
\caption{Implications of the weighted-averaging and the weighted-median mechanisms. Panel~(a) is an example of a $6\times 6$ influence matrix and the corresponding influence network with 6 nodes. Pandel~(b) illustrates the underlying implication of the weighted-averaging opinion update. Panel~(c) plots the cognitive dissonance function for node 1 in the influence network shown in Panel~(a), following the weighted-median mechanism. Node 1 updates its opinion by first sorting its social neighbors' opinions and picking the one such that the cumulative weights assigned to the opinions on its both sides are less than 0.5.}\label{fig:averaging-WM}
\end{figure*}

In this paper, we resolve the inherent unrealistic features of weighted averaging in a more fundamental way. Instead of further extending the DeGroot model, we propose the \emph{weighted-median mechanism} as an alternative micro-foundation of opinion dynamics, in which opinion attractiveness and opinion distance are not intrinsically coupled. This new mechanism is not proposed arbitrarily, but derived from network games and the cognitive dissonance theory in psychology, and is supported by an online experiment dataset. As will be manifested in later sections, such an inconspicuous change from averaging to median leads to rich consequences. The weighted-median mechanism broadens the applicability of opinion dynamics models to multiple-choice issues with ordered discrete options, e.g., political elections. Moreover, comparative numerical studies indicate that, the weighted-median mechanism, despite being the simplest in form, captures non-trivial real-world features of opinion evolution, while some widely-studied extensions of the DeGroot model fail to, in the following aspects: First, the weighted-median mechanism predicts that consensus is less likely to be achieved in larger groups or groups with more clustered network structure; Second, the weighted-median mechanism predicts that extreme opinion holders tend to reside in peripheral areas of social networks and form into local clusters, which resembles the pattern revealed by a large-scale Twitter dataset; Third, the weighted-median mechanism generates various empirically observed public opinion distributions without deliberately tuning any parameters. In summary, while it is implausible for one single model to explain every aspect of real-world opinion dynamics, the clear sociological interpretations, the empirical evidence, and the realistic model predictions in various aspects support the weighted-median mechanism as a well-founded and expressive micro-foundation of opinion dynamics.

\section{Results and Discussion}

\subsection{Model Derivation and Set-up}

\paragraph{Model derivation} The derivation of the weighted-median mechanism is inspired by network games and the \emph{cognitive dissonance} theory in psychology: Individuals experience cognitive dissonance by disagreeing with others and tend to reduce the dissonance by adjusting their opinions~\cite{LF:1957,DCM-WW:05}. Such dissonance can be mathematically formalized in different ways~\cite{PG-JL-FS:14} and the arguably most parsimonious form is
\begin{align*}
u_i(x_i,x_{-i}) =\! \sum_{j:\, w_{ij}>0}\! w_{ij}|x_i-x_j|^{\alpha},\text{ for any individual }i,
\end{align*} 
where $x_{-i}$ denotes the opinions of all the other individuals except $i$, and $\alpha>0$ is an important model parameter. In this context, individuals' opinion updates can be modeled as the following best-response dynamics: for any $i$,
\begin{equation}\label{eq:best-response-dyn}
x_i(t+1) \in \argmin_{z\in \mathbb{R}} \sum_{j:\, w_{ij}>0} w_{ij} |z-x_j(t)|^{\alpha}.
\end{equation} 
Although it might be overly assertive to claim that such ``dissonance functions'' really exist and are being minimized in human minds, the above framework does help derive opinion-update mechanisms with clear sociological interpretations. For example, due to the convexity of $x^{\alpha}$ for $x\ge 0$, $u_i(x_i,x_{-i})$ with $\alpha\!>\!1$ implies that moving towards a distant opinion reduced more dissonance than moving towards a nearby opinion by the same distance. Namely, distant opinions are more attractive. In particular, $\alpha\!=\!2$ results in the DeGroot model~\cite{DB-JK-SO:15}. On the other hand, $\alpha\!<\!1$ implies that nearby opinions are more attractive. In this letter, we adopt the neutral hypothesis $\alpha\!=\!1$, which does not imply any pre-assumption on how opinion attractiveness is coupled with opinion distance. If necessary, one could incorporate any such coupling by assuming opinion-dependent weights $w_{ij}(x)$, which is a formidable research direction but out of the scope of this letter. In turns out that equation~\eqref{eq:best-response-dyn} with $\alpha=1$ derives the weighted-median mechanism, illustrated in Fig.~\ref{fig:averaging-WM}(c) and formalized below. The detailed derivation is given by Supplementary Section II.2.

\paragraph{Model set-up} The weighted-median model is formalized as a discrete-time stochastic process. Given the influence matrix $W=(w_{ij})_{n\times n}$ and the initial condition $x(0)\in \mathbb{R}^n$, at each time $t+1$, an individual $i$ is randomly activated and update their opinion via the following weighted-median mechanism:
\begin{equation}\label{eq:WM-dyn}
x_i(t+1) = \Med_i\big(x(t);W\big),
\end{equation}
where $\Med_i\big(x(t);W\big)$ denotes the \emph{weighted median} of the n-tuple $x(t) = \big( x_1(t),x_2(t),\dots,x_n(t) \big)$ associated with the weights $(w_{i1},w_{i2}\dots,w_{in})$, i.e., the $i$-th row of the matrix $W$. The value of $\Med_i\big(x(t);W\big)$ is in turn given as follows: $\Med_i\big(x(t);W\big)=x^*\in \mathbb{R}$ if $x^*$ satisfies 
\begin{align*}
\sum_{j:\,x_j<x^*}w_{ij}\le \frac{1}{2},\quad \text{and} \quad \sum_{j:\, x_j>x^*} w_{ij}\le \frac{1}{2}.
\end{align*}
For generic weights $W$, $\Med_i\big(x(t);W\big)$ is unique. Otherwise, let $\Med_i\big(x(t);W\big)$ be the weighted median closest to $x_i(t)$, which again guarantees its uniqueness, see the Supplementary Section II.1 for a detailed discussion.

\paragraph{Broader applicability of the weighted-median model} The weighted-median operator is well-defined as long as opinions are ordered. This prominent feature broadens the applicability of opinion dynamics models to multiple-choice issues with discrete and ordered options, which have not been extensively studied before by quantitative models. Debates and decisions about ordered multiple-choice issues are prevalent in reality. For example, in modern societies, many political issues are evaluated along one-dimension ideology spectra and political solutions often do not lend themselves to a continuum of viable choices. At a fundamental level, the weighted-median mechanism is independent of numerical representations of opinions. Such representations may be non-unique and artificial for any issue where the opinions are not intrinsically quantitative.  Obviously, a nonlinear opinion rescaling leads to major changes in the evolution of the averaging-based opinion dynamics. It is notable that the human mind often perceives and manipulates quantities in a nonlinear fashion, e.g., the perception of probability according to prospect theory~\cite{DK-AT:79}. 

\subsection{Empirical Validation}

The weighted-median mechanism, derived from psychological theory and first principles, is also supported by empirical evidence. Analysis of an online experiment dataset~\cite{CVK-SM-PG-PJR-JMH-VDB:16} indicates that median-based mechanisms enjoy significantly lower errors than averaging-based mechanisms in predicting individuals' opinion shifts under social influence. In each such experiment, 6 anonymous individuals answer 30 questions sequentially within tightly limited time. The questions are guessing the number of dots in a certain color in a given image, see Fig.~\ref{fig:empirical}(a) for one example. For each question, the 6 participants answer for 3 rounds. After each round, they see all the 6 participants' answers anonymously as feedback and possibly alter their own answers. The dataset records the participants' answers in each round of the 30 questions. Such experiment design has several desired features. Firstly, the questions being asked can be considered as judgmental issues, since there is no systematic way to solve them in limited time but subjective guessing. Secondly, since the participants see each other's answers anonymously, the underlying influence network is conceivably all-to-all with uniform weights. Namely, the experiment design rules out any other factor, e.g., prejudice or communication pattern, but focuses on the core comparison between median and average.

\begin{figure*}[htb]
\centering
\includegraphics[width=0.90\linewidth]{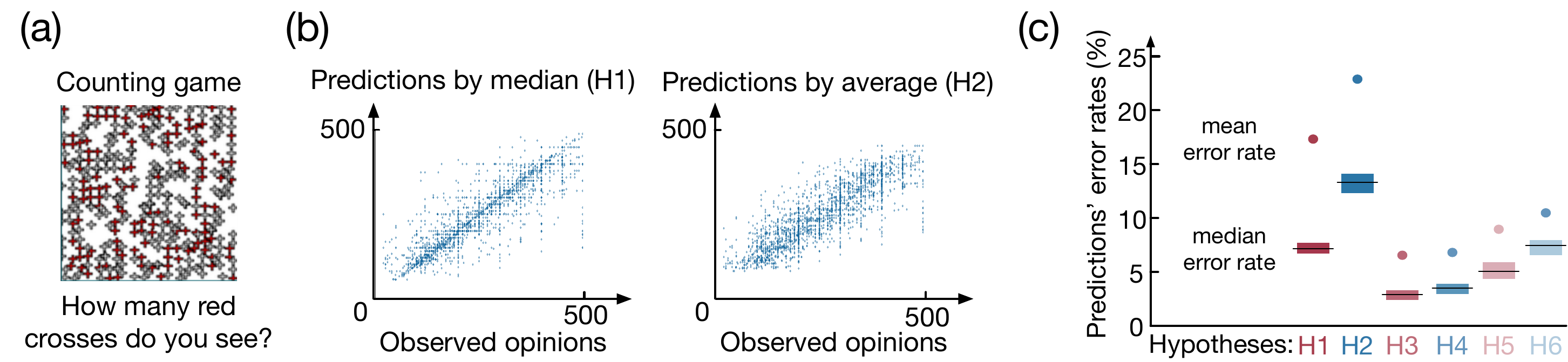}
\caption{Empirical analysis of the experiment dataset~\cite{CVK-SM-PG-PJR-JMH-VDB:16}. Panel~(a) shows an example of the counting game. Panel~(b) shows the scatter plots between the participants' observed 3rd-round answers and the predictions by median (hypothesis H1) and average (hypothesis H2) respectively. Panel~(c) is a visualized presentation of some indicative statistics of hypothese H1-H6's prediction errors. The black bars indicate the medians of prediction error rates for each hypothesis, while the vertical ranges of the colored rectangles are the associated 95\% confidence intervals, computed by the \emph{binomial distribution method}~\cite{MB:15}. The colored dots correspond to the means of the prediction error rates for each hypothesis.}\label{fig:empirical}
\end{figure*}

We randomly sampled 18 experiments from the dataset, in which 71 participants answer all the 30 questions at each round. For each question, we predict the participants' third-round answers based on their second-round answers using the following hypotheses H1-H6 in pairs: Each participant $i$'s answer $x_i(t+1)$ at the $(t+1)$-th round is given by
\hspace{-0.3cm}
\begin{align*}
\textup{H1: }\,\,  x_i(t+1)  &=\textup{Median}\big(x(t)\big);\\
\textup{H2: }\,\,  x_i(t+1)  &=\textup{Average}\big(x(t)\big);\\
\textup{H3: }\,\,  x_i(t+1)  &=\gamma_i(t) x_i(t) + (1-\gamma_i(t))\textup{Median}\big(x(t)\big);\\
\textup{H4: }\,\,  x_i(t+1)  &=\beta_i(t) x_i(t) + (1-\beta_i(t))\textup{Average}\big(x(t)\big);\\
\textup{H5: }\,\,  x_i(t+1)  &=\tilde{\gamma}_i(t) x_i(1) + (1-\tilde{\gamma}_i(t))\textup{Median}\big(x(t)\big);\\
\textup{H6: }\,\,  x_i(t+1)  &=\tilde{\beta}_i(t) x_i(1) + (1-\tilde{\beta}_i(t))\textup{Average}\big(x(t)\big),
\end{align*}
where ``Median'' and ``Average'' means arithmetic median and average of all the six participants' answers, respectively. If there are two arithmetic medians, then Median$(x(t))$ denotes the one closest to $x_i(t)$. Hypothesis H3 (H4 resp.) can be interpreted as the median (averaging resp.) mechanism with ``inertia'', while hypothesis H5 (H6 resp.) can be interpreted as the median (averaging resp.) mechanism with ``prejudice''. For hypotheses H3-H6, the parameters $\gamma_i(t)$, $\beta_i(t)$, $\tilde{\gamma}_i(t)$ and $\tilde{\beta}_i(t)$ are estimated by least-square linear regression based on the participants' answers in the first 20 questions as the training set. Then these estimated parameters are used to predict their answers in the remaining 10 questions. 

Using the above method for $t=2$, we obtain 71$\times$30 = 2130 predictions of the participants' 3rd-round answers by each of H1 and H2, and 71$\times$10=710 predictions by each of H3-H6. Fig.~\ref{fig:empirical}(b) shows the scatter plots between the observed answers and the predictions by H1 and H2. We compute the error rate for each prediction by H1-H6 as follows:
\begin{align*}
\text{error rate} = \frac{\,|\,\text{predicted value}-\text{observed value}\,|\,}{\text{observed value}}
\end{align*}
Some indicative statistics of the prediction error rates for H1-H6 are visualized in Fig.~\ref{fig:empirical}(c) and are presented in details in Supplementary Fig.~1, according to which the median error rate of the predictions by median (H1) is 46.36\% lower than that of the predictions by average (H2). In addition, for each pair of hypotheses, the median-based mechanism bears lower median (and also mean) prediction error rate than the average-based counterpart. Notably, hypotheses H3 and H4 achieve remarkably low prediction errors by introducing individual inertia as additional parameters. Despite being useful for fitting the models, these parameters do not reflect intrinsic attributes of the individuals, nor are they stable over time. Hence, we refrain from such extensions and focus on the core issue, namely mean v.s. median. In addition, we also predict the participants' opinion shifts from the first round to the second round of each question. The results yield quantitatively similar conclusions, see the Supplementary Fig.~1. 

\subsection{Comparative numerical studies}

\begin{figure}[htb]
\centering
\includegraphics[width=1\linewidth]{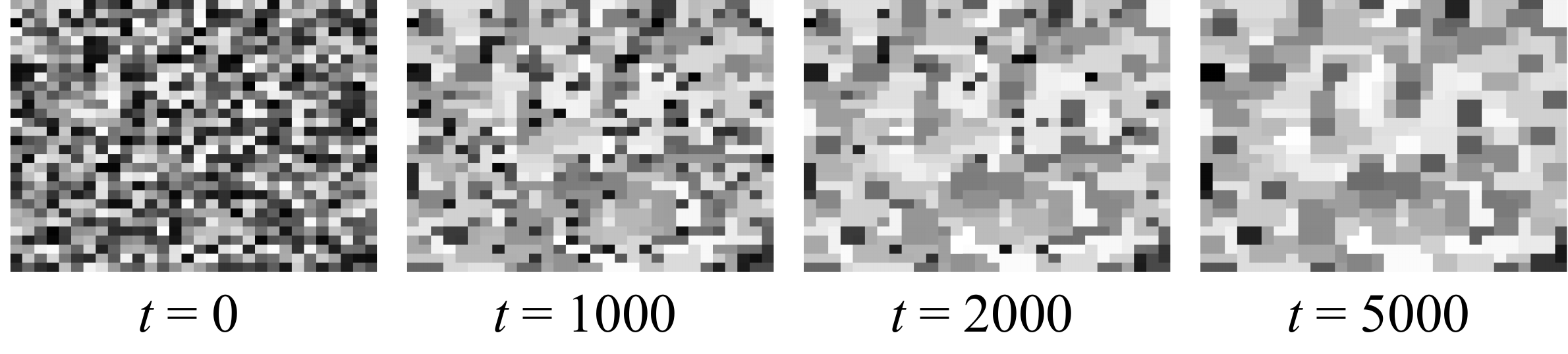}
\caption{One simulation of the weighted-median model on a $30\times 30$ lattice graph. Each block is an individual and is bilaterally connected with all their adjacent blocks (not including the diagonally adjacent blocks). Each individual has a self loop and uniformly assigns weights to all their neighbors including themself. Individuals' initial opinions are independently randomly generated according to the uniform distribution on $[-1,1]$. The grayscale of each block is proportional to the absolute value of the individual's final opinion, i.e., their degree of extremeness. After 5000 time steps, the evolution reaches an equilibrium.}\label{fig:eg-lattice}
\end{figure}

Fig.~\ref{fig:eg-lattice} shows a typical evolution of the weighted-median model on a lattice graph, from which some immediate observations can be obtained. First, unlike the DeGroot model, individuals in the weighted-median model do not always reach consensus but usually form into different opinion clusters. Second, most of the extreme opinion holders (i.e., the dark grey blocks), initially scattered uniformly in the lattice, gradually convert to more moderate opinions. Namely, the typical effect of social influence on moderating the opinions of individuals in groups are still present but not overly strong as in the DeGroot model. 

Further insights revealed by the weighted-median model are to be presented in the rest of this section. Particularly, we compare the behavior of the weighted-median model with some widely-studied extensions of the DeGroot model, including the Friedkin-Johnsen model~\cite{NEF-ECJ:90}, the biased-assimilation model~\cite{PD-AG-DTL:13}, and the networked bounded-confidence model~\cite{RP-MF-BT:19}, all with randomized model parameters. Their mathematical forms and simulation set-ups are provided in the Methods section. 

\paragraph{Consensus probability} Since the weighted-median mechanism resolves the overly large attractions between distant opinions, the effects of network structures on generating persistent disagreement naturally emerge. We investigate how the group size and the clustering coefficient of the underlying influence network affect a group's probability of reaching consensus. We simulate different models on Watts-Strogatz small-world networks~\cite{DJW-SHS:98}, whose structure is tuned by three model parameters: the network size $n$, the average degree $d$, and the \emph{rewiring probability} $\beta$. Specifically, the smaller $\beta$, the more clustered the network is. For the simulation results shown in Fig.~\ref{fig:consensus_prob}(a)-(c), we fix the rewiring probability as $\beta=1$ and estimate how the probability of reaching consensus changes with the network size $n$, under various fixed values of the average degree $d$. For the simulation results shown in Fig.~\ref{fig:consensus_prob}(d)-(f), we fix the network size as $n=30$ and estimate how the probability of reaching consensus changes with the rewiring probability $\beta$, under various fixed values of the average degree $d$. For each model and network set-up, the consensus probability is estimated over 5000 independent simulations. As indicated by Fig.~\ref{fig:consensus_prob}(a)(d), in the weighted-median model, consensus is less likely to be achieved in larger or more clustered networks. This feature is consistent with previous empirical studies~\cite{APH:52,YY-MA:01} and even everyday experience. Predictions by other models are shown in Fig.~\ref{fig:consensus_prob}(b)(c)(e)(f): The Friedkin-Johnsen model almost surely leads to non-consensus; The biased assimilation model and the networked bounded-confidence model capture the decreasing of consensus probability with network size, but does not show clear patterns regarding the relation between consensus probability and clustering coefficient.

\begin{figure*}[htb]
\begin{center}
\includegraphics[width=1\linewidth]{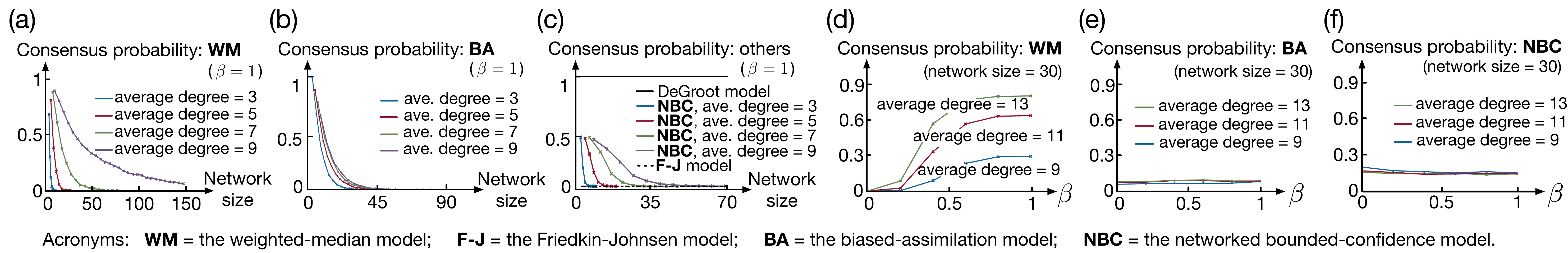}
\caption{Different models' predictions on how consensus probability depends on network size and clustering coefficient. These models are simulated on Watts-Strogatz small-world networks~\cite{DJW-SHS:98}. In Panel~(a)-(c), we fix the average individual degree $d$ and the rewiring probability $\beta$, and plot how the consensus probability changes with the network size $n$. Panel~(a) presents the predictions by the weighted-median model when $\beta=1$. For other values of $\beta$, the results are qualitatively similar, e.g., see Supplementary Fig.~2 for the results when $\beta=0.3$. In Panel~(d)-(f), we fix $n$ and $d$, and plot how the consensus probability changes with $\beta$. Panel~(d) presents the predictions by the weighted-median model when $n=30$. For other values of $n$, the results are qualitatively similar, e.g., see Supplementary Fig.~2 for $n=20$. Since the DeGroot and the Friedkin-Johnsen models lead to trivial predictions of either almost-sure consensus or almost-sure disagreement, their curves are not plotted in Panel~(d)-(f). 
}
\label{fig:consensus_prob}
\end{center}
\end{figure*} 

\paragraph{Locations of extreme opinions} From Fig.~\ref{fig:eg-lattice}, one could already see that extreme opinions in the lattice graph behave differently than moderate opinions. To further investigate how extreme opinions are distributed in social networks, we simulate different models for 100 times independently on randomly generated scale-free networks~\cite{ALB-RA:99} with 5000 nodes. The initial opinions are uniformly randomly generated from $[-1,1]$ and opinions are classified into 4 categories, see Fig.~\ref{fig:extreme-peripheral}(a). We estimate the in-degree centrality distributions for individuals holding different categories of opinions at the steady states of each simulation. As Fig.~\ref{fig:extreme-peripheral}(b) indicates, only in the weighted-median model, the in-degree distribution curves for different categories of opinions are clearly separated, and, moreover, the curve for extreme opinions decays the fastest as in-degree increases. That is, only the weighted-median model shows that extreme opinions tend to reside in peripheral areas of social networks. This feature is consonant with previous empirical, conceptual, and case studies~\cite{JW-JW-JG-TM:74,CM-SM:08,JRH-AKW:12,ECH:13,ETM-RWM:14,SLP-MJG-HM-MF-MVE:15}, which explain opinion radicalization via social-influence processes and identify social marginalization as a key cause. Such a connection has barely been captured by quantitative opinion dynamics models and the weight-median mechanism provides perhaps the simplest explanation for it. To avoid the risk of bias due to the higher probability of being absolutely stubborn (self-weight $>1/2$) in the weighted-median model when the in-degree is small, we perform a second experiment on graphs without self-weights, and obtained similar results, see Supplementary Fig.~4. Simulations for closeness and between centrality or for different categorizations of opinions, also lead to similar results and are presented in Supplementary Fig.~3 and Supplementary Fig.~5. 

\begin{figure*}[htb]
\begin{center}
\includegraphics[width=1\linewidth]{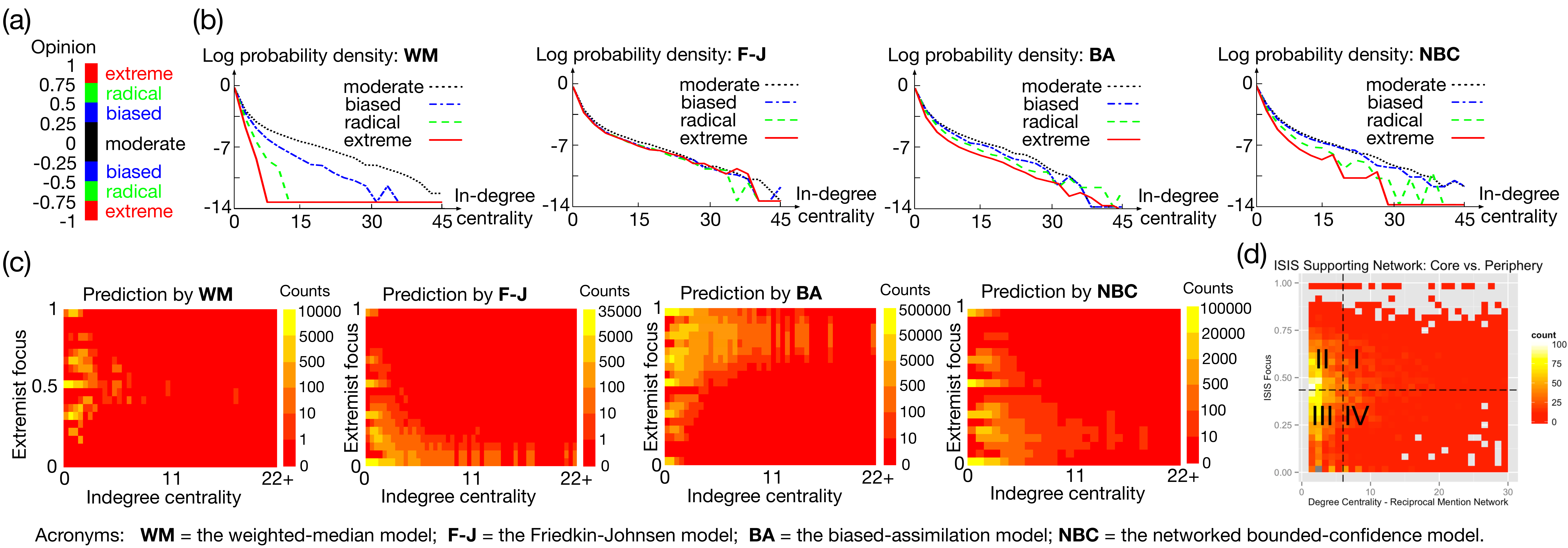}
\caption{Distributions of extreme opinions predicted by different models. Panel~(a) is the categorization of opinions. Panel~(b) shows different models' predictions on the in-degree centrality distributions for individuals holding different categories of opinions at the steady states. Panel~(c) shows different models' predictions on the two-dimensional distributions, i.e., the in-degree and the extremist focus, for the extreme opinion holders at steady states. In each heat map, the last column ``22+'' records the number of extreme individuals with in-degrees larger than or equal to 22. Panel~(d) is Figure 5 in~\cite{MCB-KJ-MC:17}, licensed under Creative Commons CC0 public domain dedication (CC0 1.0). This figure plots the empirical distribution of randomly sampled Twitter users over the in-degree and the ISIS focus (the ratio of one's pro-ISIS social neighbors).
}
\label{fig:extreme-peripheral}
\end{center}
\end{figure*}  

For each model in comparison, we further simulate them on a scale-free network with 2000 nodes for 1000 times independently. To avoid the trivial cases that some individuals might stick to extreme opinions just because they have self loops with weights larger than 1/2, the simulations are conducted on networks without self loops. For extreme opinion holders at the steady states, we compute their \emph{extreme foci}, i.e., the ratios of their neighbors also holding extreme opinions, and plot their two-dimensional distributions over the in-degree and the extreme focus, see the heat maps in Fig.~\ref{fig:extreme-peripheral}(c). The heat map generated by the weighted-median model exhibits a clearly distinct pattern to those generated by the other models: In the weighted-median model, extreme opinion holders tend to have low in-degrees and their extreme foci concentrate around the value 0.5, which implies that they form into local clusters in peripheral areas of the networks. This observation indicates a mechanistic explanation for opinion radicalization among socially marginalized individuals: In social networks, some local clusters are formed by individuals with low centrality, which usually implies few social contacts. Inside those local clusters, if extreme opinions constitute the ``mainstream'', i.e., the weighted-median opinions, individuals will adhere to extreme opinions by yielding to social influence, due to the overwhelming social pressure and lack of diverse information sources. 
Remarkably, the heat map generated by the weighted-median model impressively resembles a real dataset of the network among randomly sampled Twitter users, in which some users have their accounts suspended for posting pro-ISIS terrorism contents and are considered as extreme opinions holders, see Fig.~\ref{fig:extreme-peripheral}(d). 

\begin{figure*}
\centering
\includegraphics[width=1\linewidth]{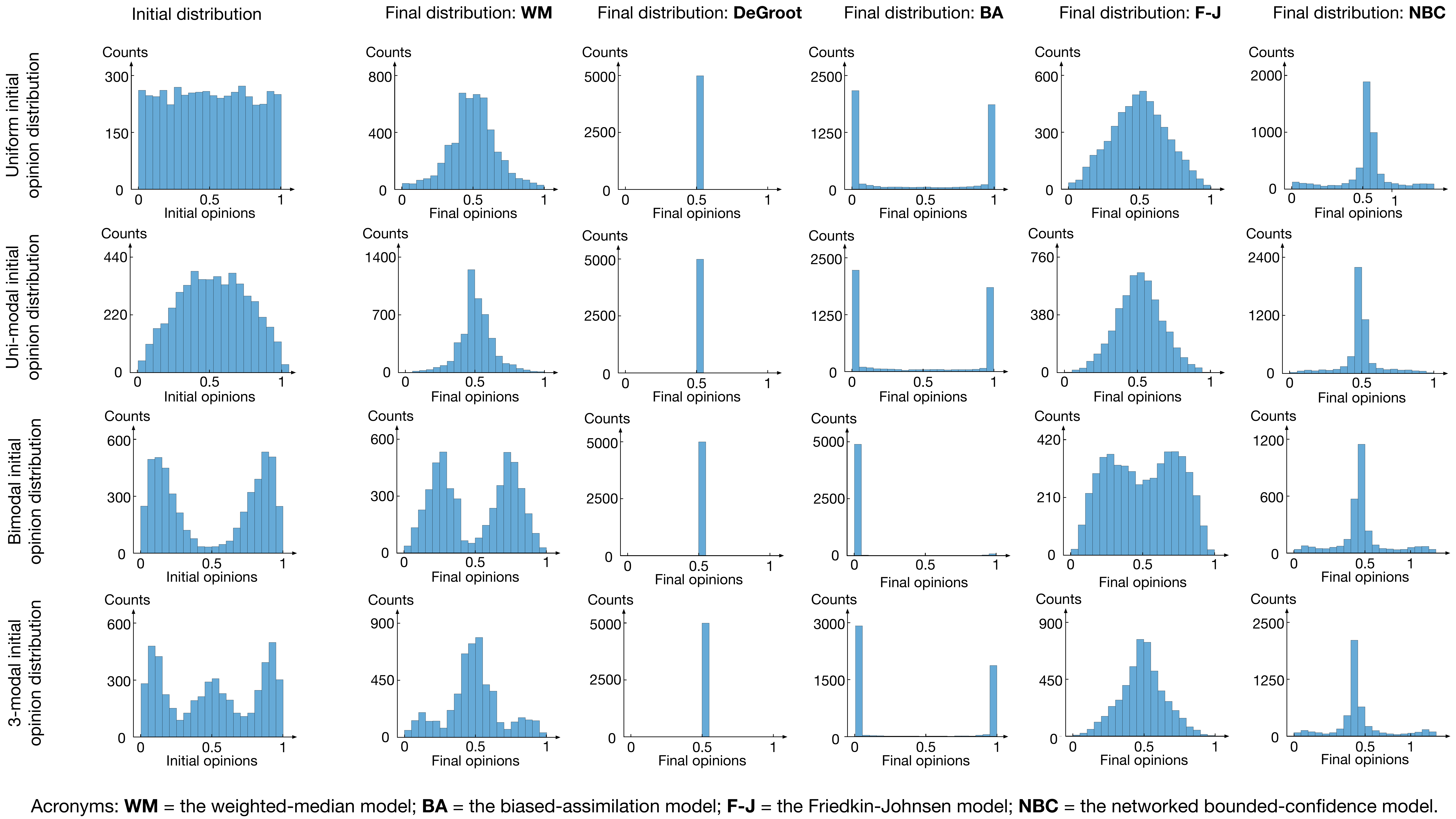}
\caption{Distributions of the initial opinions and the final opinions predicted by different models. All
simulations are run on the same scale-free network with
5000 nodes and starting with the same randomly generated initial conditions.}\label{fig:initial-final-SF}
\end{figure*}

\paragraph{Steady public opinion distributions} Empirical evidence suggests public opinions do not only achieve persistent disagreement, but also form into certain steady distributions~\cite{JL:17,KJ-JMB-JG-DS-PM:19}. In fact, it has long been an open problem what mathematical models naturally lead to the emergence of various empirically observed steady public opinion distributions~\cite{NEF:15}. By simulating different models on a randomly generated scale-free network with 5000 nodes, we compare their predictions on the final steady opinion distributions, starting from various initial opinion distributions. Fig.~\ref{fig:initial-final-SF} shows a set of typical simulation results. Among all the models in comparison, only the weighted-median model, without deliberately tuning any model parameters, naturally generates various types of empirically-observed steady distributions of public opinions. Comparisons conducted on a small-world network~\cite{DJW-SHS:98} indicate similar conclusions and are provided in Supplementary Fig.~6. Namely, the weighted-median model provides perhaps the simplest explanation of the famous Abelson's diversity puzzle~\cite{RPA:64}: \textit{“Since universal ultimate agreement is an ubiquitous outcome of a very broad class of mathematical models, … what on earth one must assume in order to generate the bimodal outcome of community cleavage studies.”}   

\subsection{Conclusions and Future Research Directions}

To sum up, with minimal assumptions, the weighted-median mechanism resolves the unrealistic proportionality between opinion attractiveness and opinion distance implied by the widely-adopted weighted-averaging mechanism. Despite its simplicity in form, the weighted-median mechanism leads to higher accuracy in quantitatively predicting individuals' opinion shifts in an online experiment and captures various interesting real-world phenomena. While it is implausible for one single model to explain every aspect of real-world opinion evolution, all the aforementioned features support the weighted-median mechanism as a well-founded and expressive micro-foundation of opinion dynamics, especially for multiple-choice issues with discrete and ordered options.

A major limitation of the weighted-median mechanism is that no new opinion is created during the opinion updates. Therefore, it does not capture the behavior that individuals compromise at intermediate opinions. This limitation could be resolved by assuming that individuals move towards instead of directly taking their weighted-median opinions. In addition, one could make the model more realistic by considering state-dependent weights, e.g., by assuming that individuals with more extreme opinions become more stubborn, as in~\cite{VA-FB-AKS:16b}. Moreover, many non-trivial extensions introduced to the classic DeGroot model can also be incorporated into the weighted-median mechanism, e.g., the presence of antagonistic relations~\cite{GS-CA-JSB:19}, individual prejudice~\cite{NEF-ECJ:90}, the logical constraints among issues~\cite{NEF-AVP-RT-SEP:16}, and the issue alignments~\cite{FB-PLS-IMS-MS:20,FB-PLS-IMS-MS:21}. In addition, rigorous analysis of the dynamic behavior, e.g., convergence and graph-theoretic conditions for consensus/disagreement, of the weighted-median model and its variations would also be of important theoretical value.

\section{Methods}

Simulations in this paper are conducted via MATLAB. In this section, we provide the information needed to replicate the simulation results presented in this paper. 

\subsection{Models in comparison}

Besides the weighted-median model, we also simulate the Friedkin-Johnsen model~\cite{NEF-ECJ:90}, the biased-assimilation model~\cite{PD-AG-DTL:13}, and the networked bounded-confidence model~\cite{RP-MF-BT:19}, their mathematical formalization and parametrization are as follows.

The Friedkin-Johnsen model~\cite{NEF-ECJ:90} assumes that individuals have persistent attachments to their initial opinions in their opinion updates. The mathematical form is written as 
\begin{align*}
x_i(t+1) = (1-a_i) \sum_j w_{ij}x_j(t) + a_i x_i(0)
\end{align*} 
with $a_i\in [0,1]$ for any individual $i$. Here the parameter $a_i$ characterizes individual $i$'s tendency of attaching to their initial opinion. In our simulations, each $a_i$ is independently randomly generated from the uniform distribution on $[0,1]$.

The biased-assimilation model~\cite{PD-AG-DTL:13} is based on the idea that individuals weight confirming evidence more heavily than  disconfirming evidence. The model is formalized as follows:
$$
\hspace{-0.13cm}x_i(t+1) = \frac{w_{ii}x_i(t)+x_i(t)^{b_i} s_i(t)}{\, w_{ii}x_i(t) + x_i(t)^{b_i} s_i(t) + (1-x_i(t))^{b_i} (d_i-s_i(t)) \,},
$$
where $s_i(t)= \sum_j w_{ij}x_j(t)$, $d_i = \sum_j w_{ij}$, and $b_i\ge 0$ is an individual parameter characterizing how biased individual $i$ is. If $b_i=0$, then individual $i$ process others' opinions according to exactly the DeGroot model; The larger $b_i$, the more heavily individual $i$ tends to weigh confirming evidence relative to disconfirming evidence, i.e., the more biased $i$ is. Since theoretical analysis indicates that $b_i=1$ is somehow a critical threshold~\cite{PD-AG-DTL:13}, we assume that each $b_i$ is independently randomly generated from the uniform distribution on $[0,2]$.   

The networked bounded-confidence model~\cite{RP-MF-BT:19} assumes that individuals are embedded on an influence network but are only influenced by those whose opinions are within certain prescribed \emph{confidence radii} from their own opinions. Its mathematical form is given as follows
\begin{align*}
x_i(t+1) = \frac{\sum_{j\in N_i:\, |x_j(t)-x_i(t)|<r_i} \,\,w_{ij} x_j(t)}{\sum_{j\in N_i:\, |x_j(t)-x_i(t)|<r_i} \,\,w_{ij}}, 
\end{align*}
for any $i$, where $N_i=\{j\,|\, w_{ij}>0\}$. Here $r_i$ is the confidence radius for individual $i$. In our simulations, if the initial opinions are generated from the range $[0,1]$ ($[-1,1]$ resp.), then the individual confidence radii are independently randomly generated from the uniform distribution on $[0,0.5]$ ($[0,1]$ resp.).

\subsection{Generation of initial opinions}

For the simulation results presented in Fig.~\ref{fig:consensus_prob} and~\ref{fig:extreme-peripheral}, the initial opinion of each individual in each simulation is independently randomly generated from the uniform distribution on $[-1,1]$. Regarding the simulation results presented in Fig.~\ref{fig:initial-final-SF}, different initial opinion distributions are generated as follows:
\begin{enumerate}
\itemsep0em
\item Regarding the uniform distribution, we let the initial opinion of each individual be independently randomly sampled from the uniform distribution on $[0,1]$, i..e, $x_i(0)\sim \textup{Unif}[0,1]$ for any $i\in \{1,\dots, n\}$;
\item Regarding the uni-modal distribution, we let the initial opinion of each individual be independently randomly sampled from the Beta distribution $\textup{Beta}(2,2)$;
\item Regarding the bimodal distribution, each individual $i$'s initial opinion is independently generated in the following way: Firstly we generate a random sample $Y$ from the Beta distribution $\textup{Beta}(2,10)$, and then let $x_i(0)=Y$ or $1-Y$ with probability 0.5 respectively;
\item Regarding the 3-modal distribution, each individual $i$'s initial opinion is independently generated in the following way: Firstly we generate two random samples $Y$ and $Z$ from $\textup{Beta}(2,17)$ and $\textup{Beta}(12,12)$ respectively, and then let $x_i(0)$ be $Y$, $1-Y$, or $Z$ with probabilities 0.33, 0.33, and 0.34 respectively. 
\end{enumerate} 
For each initial opinion distribution, we randomly generate the initial opinion of each individual independently and let the models in comparison start with the same initial condition.

\subsection{Generation of random graphs}

Simulations in this paper are conducted on either the scale-free networks or the small-world networks. All these networks consist of only bilateral links. The scale-free networks are generated according to the Barab\'{a}si-Albert preferential attachment model~\cite{ALB-RA:99}. The network construction process starts with an initial seed network, which is set as a graph with 5 nodes and with the link set $\big{\{} \{1,2\}, \{1,5\}, \{2,3\}, \{3,4\},\{4,\\5\} \big{\}}$. Whenever a new node is added to the network, two bilateral links are built according to the preferential-attachment rule. The process terminates when the number of nodes, i.e., the network size, meets the prescribed value $n$. The small-world networks are generated according to the Watts-Strogatz random graph model~\cite{DJW-SHS:98}, which involves three parameters: the network size $n$, the average degree $d$ ($d<(n-1)$), and the rewiring probability $\beta$ ($\beta\in [0,1]$).

Once a random graph is constructed via either of the above methods, unless specified, self loops are added to each node. Then each link is assigned a weight independently randomly sampled from the uniform distribution on $[0,1]$. Then the link weights are normalized so that the weights of each node's out-links (including the self loop if any) sum up to one. Namely, the corresponding adjacency matrix $W$ always satisfies $\sum_{j=1}^n w_{ij}=1$ for any $i$.

\subsection{Determination of convergence and consensus}

In the simulations of different opinion dynamics models, we adopt the following numerical criteria to determine whether a model has reached a steady state or whether the group of individuals has reached consensus. Let $t$ be the index for the iteration time step in the simulations, and let $x(t)=\big( x_1(t),x_2(t),\dots, x_n(t) \big)$ be the opinions of all the individuals after the $t$-th iteration.

For the weighted-median model, starting from $t=1$, whenever $t$ satisfies ``$t$ mod $n$ = 0'', we check whether
\begin{align*}
\lVert x(t) - x(t-n) \rVert_1 < 0.001,
\end{align*}
where $\lVert \cdot \rVert$ denotes the 1-norm of $n$-dimension vectors. If the above inequality holds for 10 consecutive checkpoints, then the model is considered as having reached a steady state and the iteration terminates. For the other models in comparison, if the inequality
\begin{align*}
\lVert x(t) - x(t-1) \rVert_1 < 0.001
\end{align*}
holds for 1000 consecutive time steps, then the model is considered as having reached a steady state and the iteration terminates.

If a model reaches a steady state at iteration time step $T$, then the individuals' opinions are considered as having reached consensus if the following inequality holds:
\begin{align*}
\left\lVert x(T) - \text{mean}\big(x(T)\big) \vectorones[n] \right\rVert_1 < 0.001,
\end{align*}
where $\vectorones[n]$ is the $n$-dimension vector with all entries equal to 1.

\newpage
\bibliographystyle{plain}
\bibliography{alias,WM,Main,new}

\textbf{Data and materials availability}: The dataset used for empirical validation in this paper are obtained from the research paper~\cite{CVK-SM-PG-PJR-JMH-VDB:16} and is available on its journal website, see\\ https://journals.plos.org/plosone/article?id=10.1371/journal.pone.0230584

\newpage

\noindent\textbf{SUPPLEMENTARY MATERIALS}
\bigskip

\setcounter{section}{0}
\renewcommand{\thesection}{S\arabic{section}}

\setcounter{equation}{0}
\renewcommand{\theequation}{S\arabic{equation}}

\setcounter{figure}{0}
\renewcommand{\thefigure}{S\arabic{figure}}

This document contains a brief review of some basic concepts in graph theory, the definition and uniqueness of weighted median, supplementary empirical results and supplementary simulations results referred to in the main text.

\section{Brief Review of Graph Theory}

Graph theory is a basic mathematical tool to model networks. Some important concepts in graph theory introduced in this section are used in the main text. A graph is a triple $G(V,E,A)$. Here $V$ denotes the set of nodes and $V = \{1, . . . , n\}$ for a network of $n$ nodes. Let $E \subseteq V \times V$ be the set of links defined as follows: $(i,j)\in E$ if there exists a link from node $i$ to node $j$. A link from node $i$ to itself is called a \emph{self loop}. For any node $i\in V$, any node $j$ with $(i,j)\in E$ is an \emph{out-neighbor} of node $i$, while any node $j$ with $(j,i)\in E$ is an \emph{in-neighbor} of node $i$. Graphs in which the links are all undirected can be considered as the graphs in which all the links are directed but bilateral. Therefore, in this document, we assume all the network links to be directed, unless specified. The graph is \emph{weighted} if a real-value weight is assigned to each link. A directed and weighted graph with $n$ nodes can be characterized by an $n\times n$ matrix $A=(a_{ij})_{n\times n}$, referred to as its \emph{adjacency matrix}. For any $i,\,j\in V$, $a_{ij}\neq 0$ if and only if there is a directed link from node $i$ to node $j$. The value of $a_{ij}$, if non-zero, denotes the weight of the link from $i$ to $j$. Since the adjacency matrix contains all the information of a graph, the graph associated with an adjacency matrix $A$ can be denoted by $G(A)$.

On a graph $G(A)$, a \emph{path} from node $i_0$ to node $i_{\ell}$ with length $\ell$ is an ordered sequence of distinct nodes $\{i_0,i_1,\dots,i_{\ell}\}$, in which $a_{i_k i_{k+1}}\neq 0$ for any $k\in \{ 0,1,\dots, \ell -1 \}$. A graph is \emph{strongly connected} if, for any $i,j\in V$, there is at least one path from $i$ to $j$. A node $i$ is a \emph{globally reachable node} if, for any $j\in V$, there exists a path from $j$ to $i$. A path from node $i$ to itself, with no repeating node except $i$, is referred to as a \emph{cycle} and the number of distinct nodes involved is called the length of the cycle. A self loop is a cycle with length 1. The greatest common divisor of the lengths of all the cycles in a graph is defined as the \emph{period} of the graph. A graph with the period equal to 1 is called \emph{aperiodic}. By definition, a graph with self loops is aperiodic.

A graph $G'(V',E')$ is a \emph{subgraph} of graph $G(V,E)$ if $V'\subseteq V$ and $E'\subseteq E\cap (V\times V)$. The subgraph $G'(V',E')$ is called an \emph{induced subgraph} of $G(V,E)$, or, equivalently, the subgraph of $G(V,E)$ \emph{induced by} $V'$, if $E'$ contains all the links in $E$ between the nodes in $V'$. A subgraph $G'$ is a \emph{strongly connected component} of $G$ if $G'$ is strongly connected and any other subgraph of $G$ strictly containing $G'$ is not strongly connected.

\section{Derivation and uniqueness of weighted median}

\subsection{Uniqueness of weighted median}
The formal definition of weighted median is given as follows:
\begin{definition}[Weighted median]\label{def:WM}
Given any $n$-tuple of real numbers $x=(x_1,\dots, x_n)$ and the associated $n$-tuple of nonnegative weights $w = (w_1,\dots,w_n)$, where $\sum_{i=1}^n w_i =1$, the \emph{weighted median} of $x$, associated with the weights $w$, is denoted by $\textup{\Med}(x;w)$ and defined as the real number $x^*\in \{x_1,\dots,x_n\}$ such that
\begin{equation*}
 \sum_{i:\,x_i<x^*} w_i \le 1/2,\quad \textup{and}\quad \sum_{i:\,x_i>x^*} w_i \le 1/2.
\end{equation*}
\end{definition}

By carefully examining this definition, one could observe that, associated with certain specific weights $w$, there might exist multiple weighted medians of $x$ satisfying the definitions above. Here we point out the following facts:
\begin{enumerate}[label={Fact~\arabic*:}]
\itemsep0em
\item The weighted median of $x$ associated with $w$ is unique if and only if there exists $x^*\in \{x_1,\dots, x_n\}$ such that
\begin{equation*}
\sum_{i:\, x_i<x^*} w_i<\frac{1}{2},\quad \sum_{i:\, x_i=x^*}w_i>0,\quad \textup{and}\quad \sum_{i:\, x_i>x^*} w_i <1/2.
\end{equation*}
In this case, $x^*$ is the unique weighted median;
\item The weighted medians of $x$ associated with $w$ are NOT unique if and only if there exists $z\in \{x_1,\dots, x_n\}$ such that $\sum_{i:\, x_i<z} w_i = \sum_{i:\, x_i\ge z} w_i = 1/2$. Among all these weighted medians of $x$, the smallest one, denoted by $\underline{x}^*$, satisfies 
\begin{equation*}
\sum_{i:\, x_i<\underline{x}^*}w_i < \frac{1}{2},\quad \sum_{i:\, x_i=\underline{x}^*}w_i >0,\quad \textup{and}\quad \sum_{i:\, x_i>\underline{x}^*}w_i = \frac{1}{2},
\end{equation*}
while the largest weighted median, denoted by $\overline{x}^*$, satisfies
\begin{equation*}
\sum_{i:\, x_i<\overline{x}^*} w_i =\frac{1}{2},\quad \sum_{i:\, x_i = \overline{x}^*} w_i >0,\quad \textup{and}\quad \sum_{i:\, x_i>\overline{x}^*}<\frac{1}{2}.
\end{equation*}
Moreover, if there exists any $\hat{x}\in \{x_1,\dots, x_n\}$ such that $\underline{x}^*<\hat{x}<\overline{x}^*$, then $\hat{x}$ is also a weighted median and it must hold that $\sum_{i:\, x_i = \hat{x}}w_i = 0$.
\end{enumerate} 
For generic weights, e.g., if $w_1,\dots,w_n$ are independently randomly generated from some continuous probability distributions, the case in Fact 2 never occurs since almost surely there does not exist any $\theta\in \{1,\dots, n\}$ such that $\sum_{i\in \theta} w_i = 1/2$. Therefore, given generic weights $w$, the weighted median of $x$ is unique.

In order to avoid unnecessary mathematical complexity, we would like to make each individual's opinion update well-defined and deterministic. Therefore, in the weighted-median opinion dynamics, we slightly change the definition of weighted median when it is not unique according to Definition~\ref{def:WM}. Consider a group of $n$ individuals discussing some certain issue. Denote by $x_i(t)$ the opinion of individual $i$ at time $t$ and let $x(t)$ be the $n$-tuple  $\big(x_1(t),\dots, x_n(t)\big)$. The interpersonal influences are characterized by the influence matrix $W=(w_{ij})_{n\times n}$, which is entry-wise non-negative and satisfies $\sum_{j=1}^n w_{ij}=1$ for any $i\in \{1,\dots,n\}$. The formal definition of weighted-median opinion dynamics is given as follows.
\begin{definition}[Weighted-median opinion dynamics]\label{def:WM-op-dyn}
Consider a group of $n$ individuals discussing on some certain issue, with the influence matrix given by $W=(w_{ij})_{n\times n}$. The weighted-median opinion dynamics is defined as the following process: At each time $t+1$, one individual $i$ is randomly picked and update their opinion according to the following equation:
\begin{equation*}
x_i(t+1) = \Med_i\big(x(t);W\big),
\end{equation*}
where $\Med_i(x(t);W)$ is the weighted median of $x(t)$ associated with the weights given by the $i$-th row of $W$, i.e., $(w_{i1},w_{i2},\dots, w_{in})$. $\Med_i\big(x(t);W\big)$ is well-defined if such a weighted-median is unique. If the weighted-median is not unique, then let $\Med_i\big(x(t);W\big)$ be the weighted median that is the closest to $x_i(t)$.
\end{definition}
This set-up guarantees the uniqueness of $\Med_i(x;W)$ since only one of the following 3 cases can occur when the weighted medians are not unique: 
\begin{enumerate}[label={\roman*)}]
\itemsep0em
\item $x_i\le \underline{x}^*$, where $\underline{x}^*$ is the smallest weighted median of $x$ associated with the weights $(w_1,\dots, w_n)$. In this case, $\Med_i(x;W)=\underline{x}^*$ is unique;
\item $x_i\ge \overline{x}^*$, where $\overline{x}^*$ is the largest weighted median of $x$ associated with the weights $(w_1,\dots, w_n)$. In this case, $\Med_i(x;W)=\overline{x}^*$ is unique;
\item $\underline{x}^*<x_i<\overline{x}^*$. According to Fact~2 for the weighted median in last paragraph, this must imply that $\sum_{j:\, x_j= x_i} w_{ij}=0$ and $x_i$ is also a weighted median of $x$ associated with the weights $(w_1,\dots, w_n)$. Therefore, in this case, $\Med_i(x;W)=x_i$ is also unique.
\end{enumerate}

Note that, if the entries of $W$ are randomly generated from some continuous distributions, then, for any subset of the links on the influence network $G(W)$, the sum of their weights is almost surely not equal to $1/2$. As a consequence, the weighted median for each individual at any time is almost surely unique. Therefore, for generic influence networks,  the weighted-median opinion dynamics defined by Definition~\ref{def:WM-op-dyn} follow a simple rule and is consistent with the formal definition of weighted median given in Definition~\ref{def:WM}.  In the rest of this article, by weighted-median opinion dynamics, or weighted-median model, we mean the dynamical system described by Definition~\ref{def:WM-op-dyn}. According to Definition~\ref{def:WM-op-dyn}, for any given initial condition $x(0) = (x_{0,1},\dots,x_{0,n})^{\top}$, the solution $x(t)$ to the weighted-median opinion dynamics satisfies $x_i(t)\in \{x_{0,1},\dots,x_{0,n}\}$ for any $i\in \{1,\dots, n\}$ and any $t\ge 0$. Moreover, according to Definition~\ref{def:WM-op-dyn}, for each node $i$, 
\begin{align*}
x_i(t+1) & >x_i(t) \quad \textup{if and only if}\quad \sum_{j:\, x_j(t)>x_i(t)} w_{ij}>1/2,\qquad \text{and}\\
x_i(t+1) & <x_i(t) \quad \textup{if and only if}\quad \sum_{j:\, x_j(t)<x_i(t)} w_{ij}>1/2.
\end{align*}

\subsection{Derivation of the weighted-median mechanism from the absolute-value cognitive dissonance function}

Consider an influence network $G(W)$ with $n$ individuals. Given the opinion vector $x$, each individual $i$'s cognitive dissonance generated by disagreeing with others can be modelled as
\begin{align*}
C_i(x_i,x_{-i}) = \sum_{j=1}^n w_{ij} |x_i-x_j|^{\alpha},
\end{align*}
and individual $i$'s opinion update can be modelled as the best response to minimize the cognitive dissonance $C_i(x_i,x_{-i})$. That is, the updated opinion of individual $i$, denoted by $x_i^+$, satisfies
\begin{equation}\label{eq:best-response}
x_i^+ = \argmin_{z\in \mathbb{R}} \sum_{j=1}^n w_{ij}|z-x_j|^{\alpha}.
\end{equation}
We use equality here in the sense that the right-hand side of the equation above is unique for generic weights $w_{ij}$'s. The following proposition states the relation between the system given by equation~\eqref{eq:best-response} and the weighted-median opinion update, when we set the value of the parameter $\alpha=1$.
\begin{proposition}[Weighted-median update as best-response dynamics]\label{prop:best-response-WM}
Given the row-stochastic influence matrix $W=(w_{ij})_{n\times n}$ and the vector $x = \big(x_1,\dots, x_n  \big)^{\top}$, the following statements holds: for any $i\in \{1,\dots,n\}$, 
\begin{enumerate}[label=\roman*)]
\itemsep0em
\item If there exists $x^*\in \{x_1,\dots, x_n\}$ such that
   \begin{equation*}
   \sum_{j:\, x_j <x^*} w_{ij} < \frac{1}{2},\quad \textup{and}\quad \sum_{j:\, x_j>x^*} w_{ij} < \frac{1}{2},
   \end{equation*}
then 
   \begin{equation*}
   \Med_i(x;W) = x^* = \argmin_z \sum_{j=1}^n w_{ij} |z-x_j|;
   \end{equation*}
\item If there does not exist such $x^*$, then the set
   \begin{equation*}
   M_i(x;W) = \Big\{ y\in \{x_1,\dots, x_n\}\,\Big|\, \sum_{j:\, x_j \le y} w_{ij} \le \frac{1}{2},\quad \sum_{j:\, x_j>y} w_{ij} \le \frac{1}{2} \Big\}
   \end{equation*}
is non-empty and 
\begin{align*}
\Med_i(x;W)  = \argmin_{y\in M_i(x;W)} |y-x_i| \in \Big[ \inf M_i(x;W),\,\, \sup M_i(x;W) \Big] = \argmin_z \sum_{j=1}^n  w_{ij} |z-x_j|.
\end{align*}
\end{enumerate}
\end{proposition}
This proposition is a straightforward consequence of Definition~\ref{def:WM} in this document and Lemma~3.1 in the paper by Sabo et al.~\cite{KS-RS:08}.

\section{Supplementary empirical results}

In this section, we provide some supplementary empirical results on the analysis of the online experiment dataset published in the paper by Kerckhove et al.~\cite{CVK-SM-PG-PJR-JMH-VDB:16} and described in Section 2.2 of the main text. As mentioned in the main text, we predict the participants' answers at each round using the following hypotheses:
\begin{align*}
\textup{Hypo. 1 (median):}\quad  x_i(t+1) & =\textup{Median}\big(x(t)\big);\\
\textup{Hypo. 2 (average):}\quad  x_i(t+1) & =\textup{Average}\big(x(t)\big);\\
\textup{Hypo. 3 (median with inertia):}\quad  x_i(t+1) & =\gamma_i(t) x_i(t) + (1-\gamma_i(t))\textup{Median}\big(x(t)\big);\\
\textup{Hypo. 4 (average with inertia):}\quad  x_i(t+1) & =\beta_i(t) x_i(t) + (1-\beta_i(t))\textup{Average}\big(x(t)\big);\\
\textup{Hypo. 5 (median with prejudice):}\quad  x_i(t+1) & =\tilde{\gamma}_i(t) x_i(1) + (1-\tilde{\gamma}_i(t))\textup{Median}\big(x(t)\big);\\
\textup{Hypo. 6 (average with prejudice):}\quad   x_i(t+1) & =\tilde{\beta}_i(t) x_i(1) + (1-\tilde{\beta}_i(t))\textup{Average}\big(x(t)\big).
\end{align*}
The meanings of the notations and parameters in the above equations are explained in the main text. Regarding the predictions of the opinion shifts from the first round to the second round, Hypotheses 5 and 6 are equivalent to Hypotheses 3 and 4 respectively. The histograms of the errors rates of the predictions by different hypotheses on the individuals' opinions at the second (third respectively) rounds are presented in Panel~(a) (Panel~(b) respectively) of Supplementary Fig.~\ref{fig:complete-empirical-analysis}.  Some indicative statistics on the error rates of the predictions of the 2nd-round opinions by different hypotheses are presented in Panel~(c) of Supplementary Fig.~\ref{fig:complete-empirical-analysis}. Regarding the predictions of opinion shifts from the 2nd round to the 3rd round, the data analysis results are provided in Panel~(d). 

\begin{figure}
\centering
\includegraphics[width=0.99\linewidth]{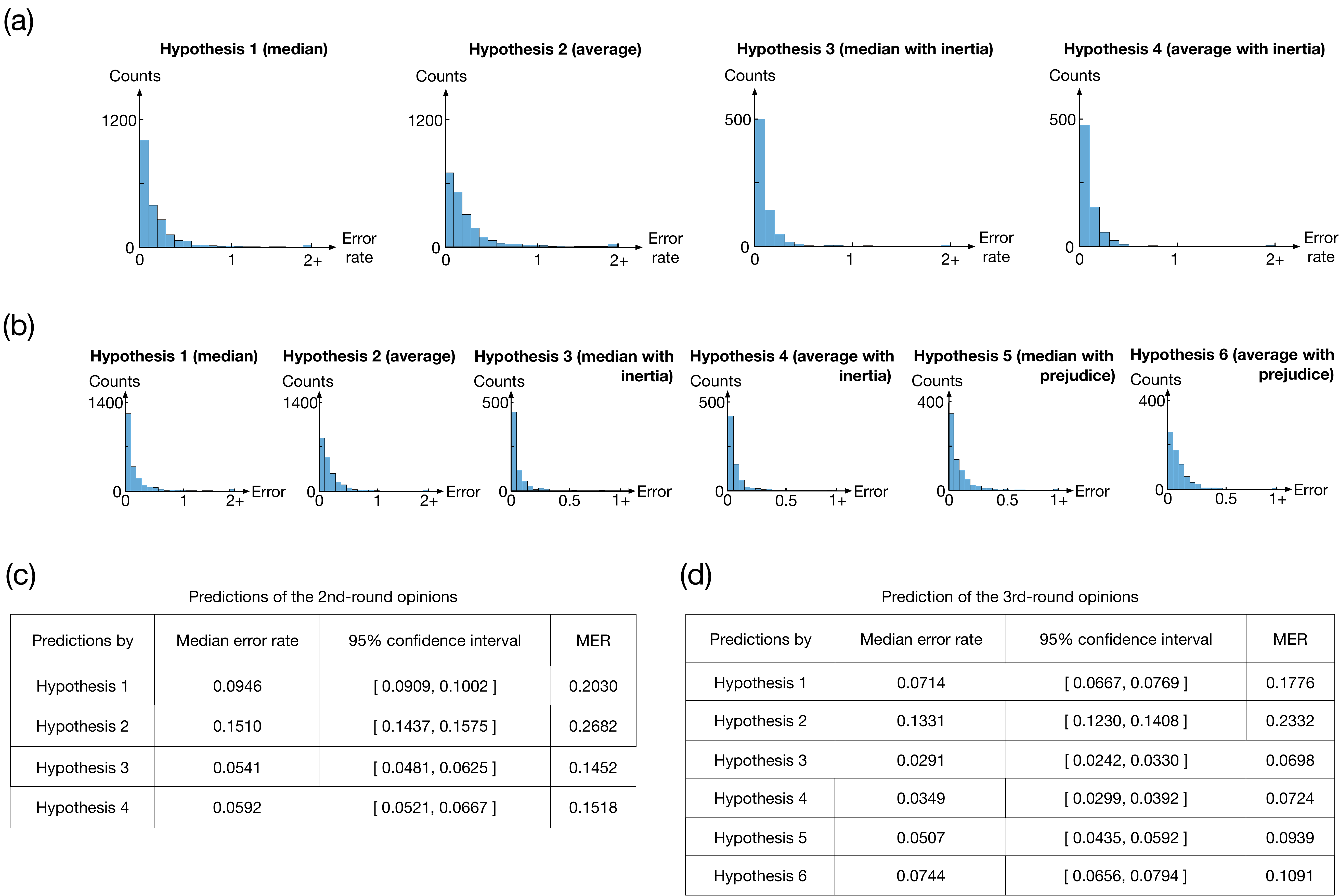}
\caption{Empirical analysis results for the dataset collected in an online human-subject experiment~\cite{CVK-SM-PG-PJR-JMH-VDB:16}. Here Hypothesis 1-6 correspond to median, average, median with inertia, average with inertia, median with prejudice, and average with prejudice, respectively. In the $x$-axis of the plots in Panel~(a) and (b), ``2+'' means ``larger than or equal to 2''. The acronym ``MAE'' in these tables is short for ``mean absolute-value error'' and ``MER'' is short for ``mean error rate''.}
\label{fig:complete-empirical-analysis}
\end{figure} 

\section{Supplementary simulation results}

\subsection{Consensus probability} 

In this main text, when we investigate the effect of network size on consensus probability, we fix the rewiring probability $\beta$ the small-world networks as $\beta=1$. The results presented in Fig.4(a)-(c) in the main text are robust to the value of $\beta$, e.g., see Supplementary Fig.~\ref{fig:SI-consensus_prob_robust}(a) for the qualitatively similar results when $\beta$ is set to be 0.3, obtained under the same simulation set-up as in the main text. Similarly, although the value of $n$ is set to be 30 in Fig.~4(d)-(f) of the main text when we investigate the effect of the rewiring probability $\beta$, the results presented there are robust to the value of $n$, e.g., see Supplementary Fig.~\ref{fig:SI-consensus_prob_robust}(b) for the qualitatively similar results when $n$ is set to be 20, obtained under the same simulation set-up as in the main text. 

\begin{figure}[htb]
\begin{center}
\includegraphics[width=0.6\linewidth]{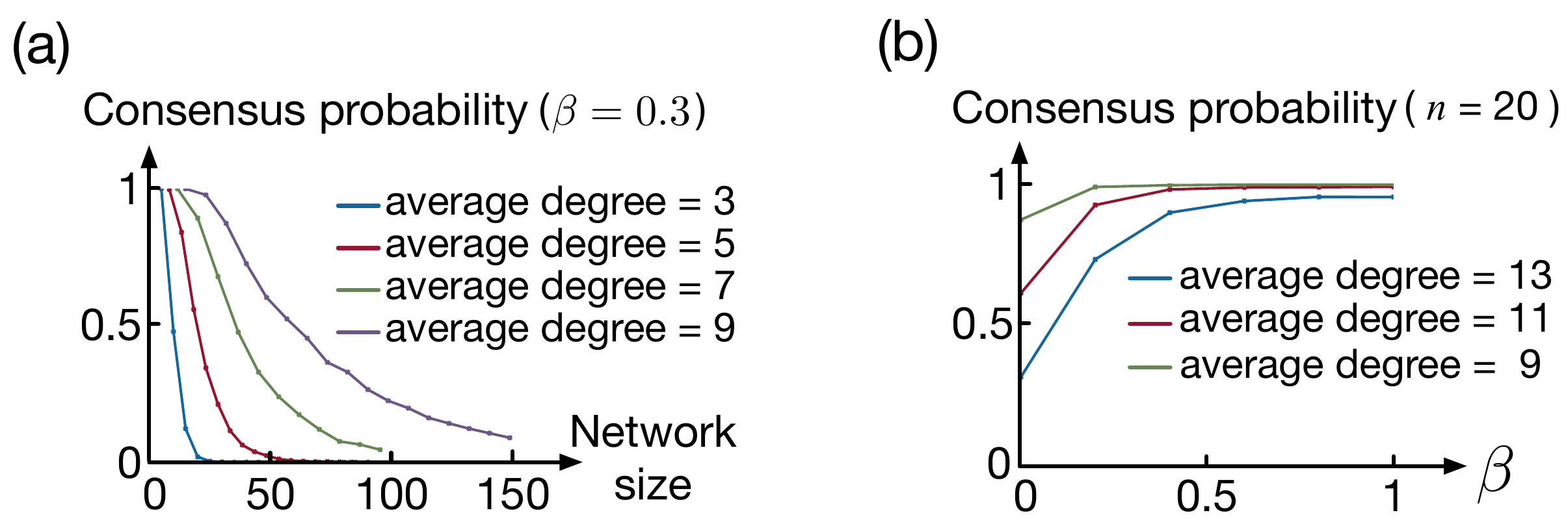}
\caption{The weighted-median model's predictions on how consensus probability depends on network size and clustering coefficient. In Panel~(a), we plot how the consensus probability changes with $n$, for fixed values of $d$ and $\beta$ ($\beta$ is set to be 0.3). In Panel~(b), we plot how the consensus probability changes with $\beta$ for fixed values of $n$ and $d$ ($n$ is set to be 20). All the probabilities are estimated over 5000 independent simulations. 
}
\label{fig:SI-consensus_prob_robust}
\end{center}
\end{figure} 

\subsection{Distribution of extreme opinions}

In Section 2.3 of the main text, we investigate the in-degree distributions for different categories of opinions. With the same simulation set-up as in the main text, we obtain qualitatively similar results if the in-degree centrality is replaced by the closeness centrality or the betweenness centrality. However, results for the eigenvector centrality do not reveal any clear pattern, see Supplmentary Fig.~\ref{fig:centrality-dist-extremeness-all}.

\begin{figure*}[htb]
\begin{center}
\includegraphics[width=0.9\linewidth]{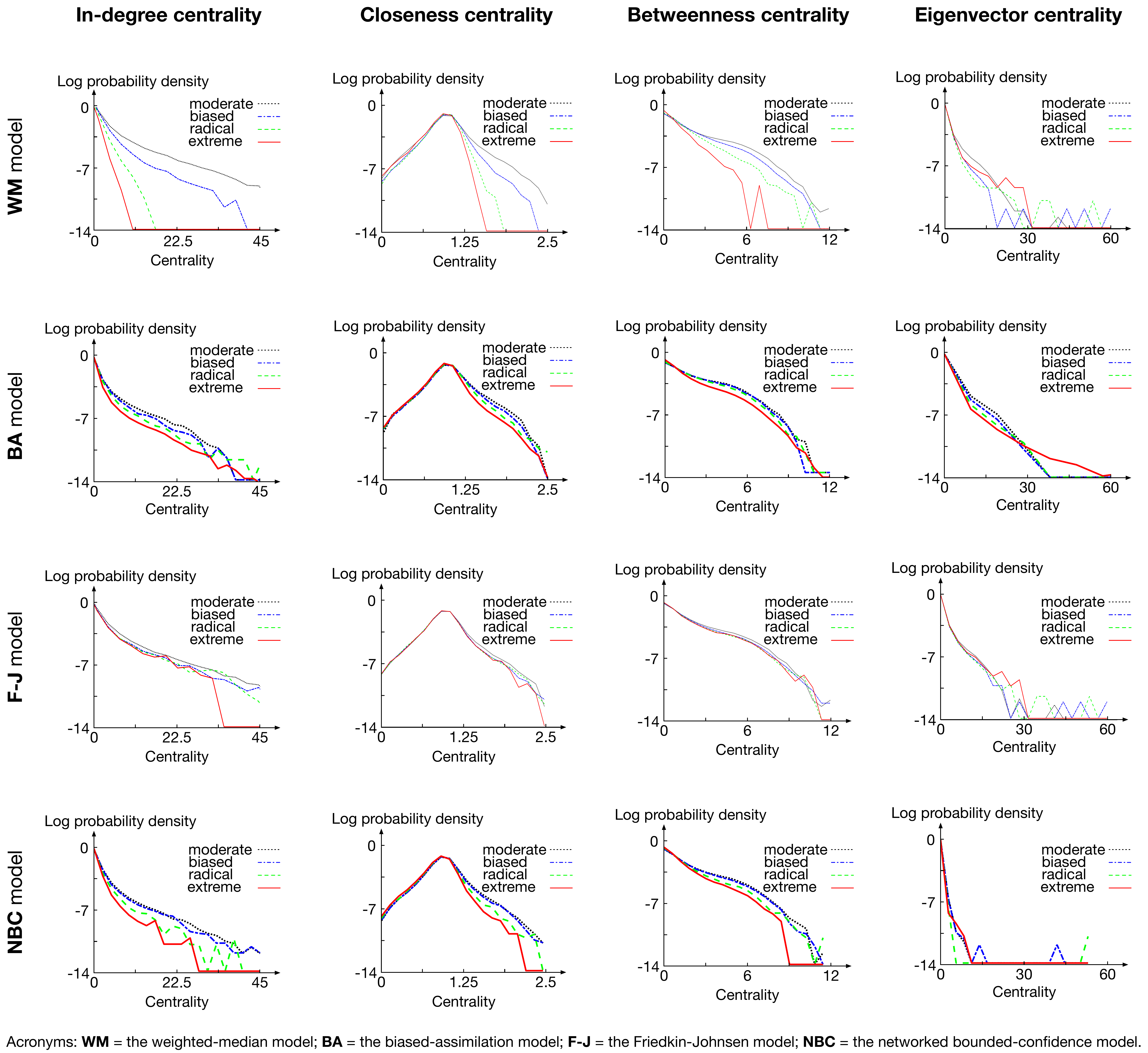}
\caption{Centrality distributions for moderate, biased, radical and extreme final opinions predicted by different models. The distributions are presented in the form of log probability density. Here the initial opinions be randomly generated from the uniform distribution $\textup{Unif}\,[-1,1]$ and classify the opinions into four categories: the \emph{moderate} opinions correspond to those in the interval $[-0.25,0.25]$; the \emph{biased} opinions correspond to those in $[-0.5,-0.25)\cup (0.25,0.5]$; the \emph{radical} opinions correspond to those in $[-0.75,-0.5)\cup (0.5,0.75]$; the \emph{extreme} opinions correspond to those in $[-1,-0.75)\cup (0.75,1]$.}
\label{fig:centrality-dist-extremeness-all}
\end{center}
\end{figure*} 

For the results shown in Fig.~5 of the main text, the opinion dynamics models in comparison are simulated on networks with self loops. Simulations with the same set-up but on networks without self loops lead to qualitatively similar results, see Supplementary Fig.~\ref{fig:centrality-dist-extremeness-no-self-loop}.
\begin{figure*}[htb]
\begin{center}
\includegraphics[width=0.99\linewidth]{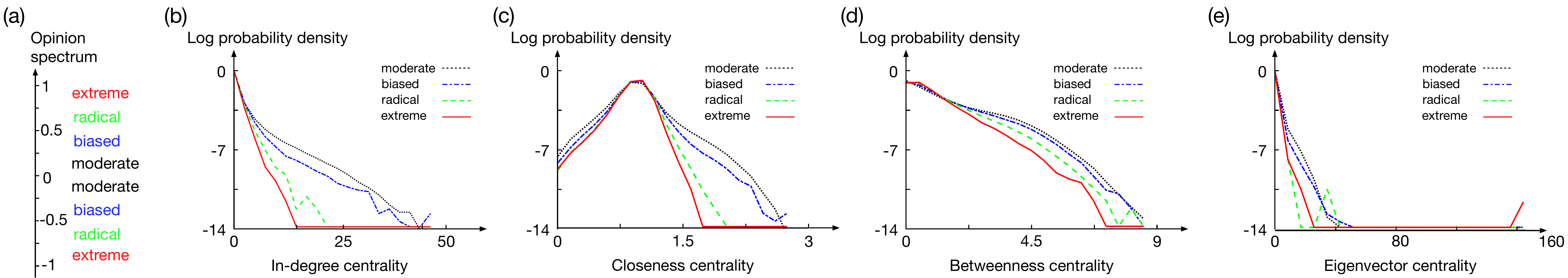}
\caption{Centrality distributions for moderate, biased, radical and extreme final opinions predicted by the weighted-median model, on a scale-free network with no self loop. The distributions are presented in the form of log probability density. The opinion spectrum is given by Panel~(a). Panels~(b)-(d) show the log probability distributions in terms of different measures of centrality.}
\label{fig:centrality-dist-extremeness-no-self-loop}
\end{center}
\end{figure*}   

To test whether the predictions by the weighted-median model on the distribution of extreme opinions, shown in Fig.~5 of the main text, are robust to the criteria of ``extreme'' opinions, we adopt two alternative classifications of opinions and repeat the simulations on the centrality distributions for different categories of opinions, as well as the extremist-focus-indegree distributions, for the weighted-median model. Qualitatively similar results are obtained, see Supplementary Fig.~\ref{fig:extreme-peripheral-alternative-criteria}.

\begin{figure}
\centering
\includegraphics[width=0.99\linewidth]{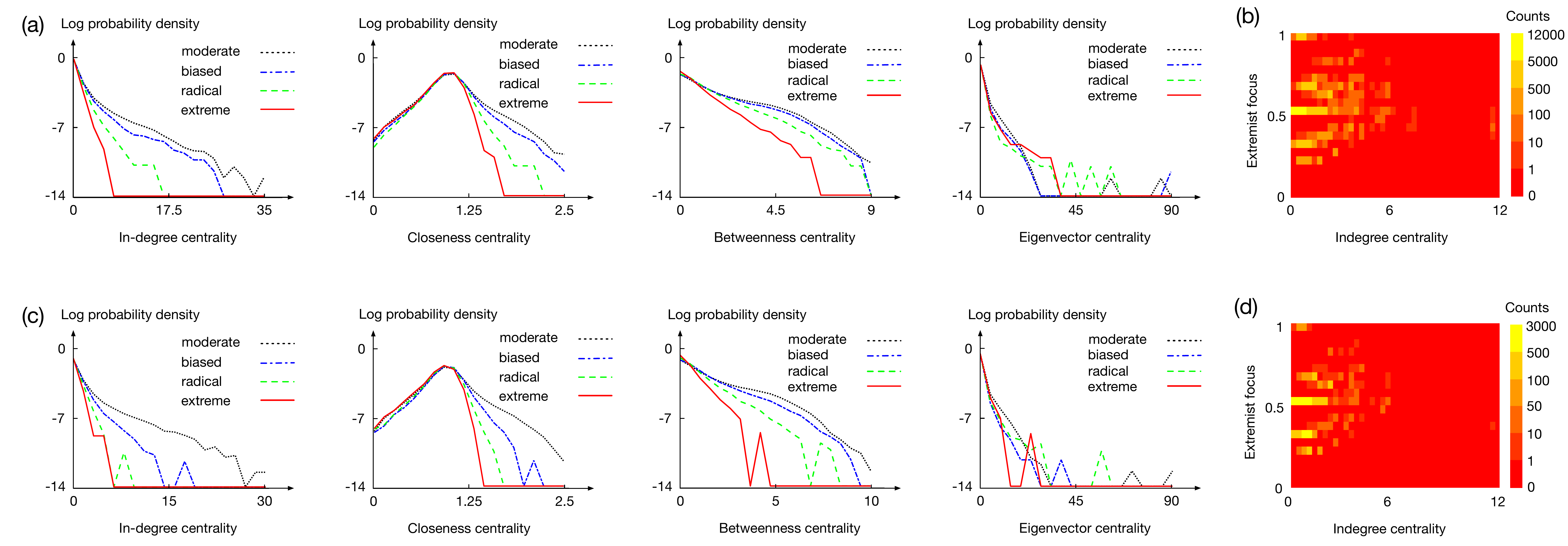}
\caption{Predictions by the weighted-median model on the distribution of extreme opinions, with different classifications of opinions. Panels~(a) and~(b) correspond to the following criteria: moderate ($[-0.2,0.2]$), biased ($[-0.4,-0.2)\cup (0.2,0.4]$), radical ($[-0.7,-0.4)\cup (0.4,0.7]$), extreme ($[-1,-0.7)\cup (0.7,1]$). Panels~(c) and~(d) correspond to the following criteria: moderate ($[-0.3,0.3]$), biased ($[-0.6,-0.3)\cup (0.3,0.6]$), radical ($[-0.9,-0.6)\cup (0.6,0.9]$), extreme ($[-1,-0.9)\cup (0.9,1]$). Panels~(a) and~(c) show the centrality distributions of different categories of opinions at the final steady states, while Panels~(b) and~(d) show the two-dimensional distributions over the extremist-focus and the in-degree centrality for the etreme opinion holders at the final steady states. }\label{fig:extreme-peripheral-alternative-criteria}
\end{figure}

\section{Steady final opinion distributions}

Fig.~6 in the main text shows the predictions on the final opinion distributed by different opinion dynamics models simulated on a scale-free network with 5000 nodes. Simulations on small-world networks with 5000 nodes lead to similar results and are presented in Supplementary Fig.~\ref{fig:initial-final-WS}.

\begin{figure*}[htb]
\centering
\includegraphics[width=0.99\linewidth]{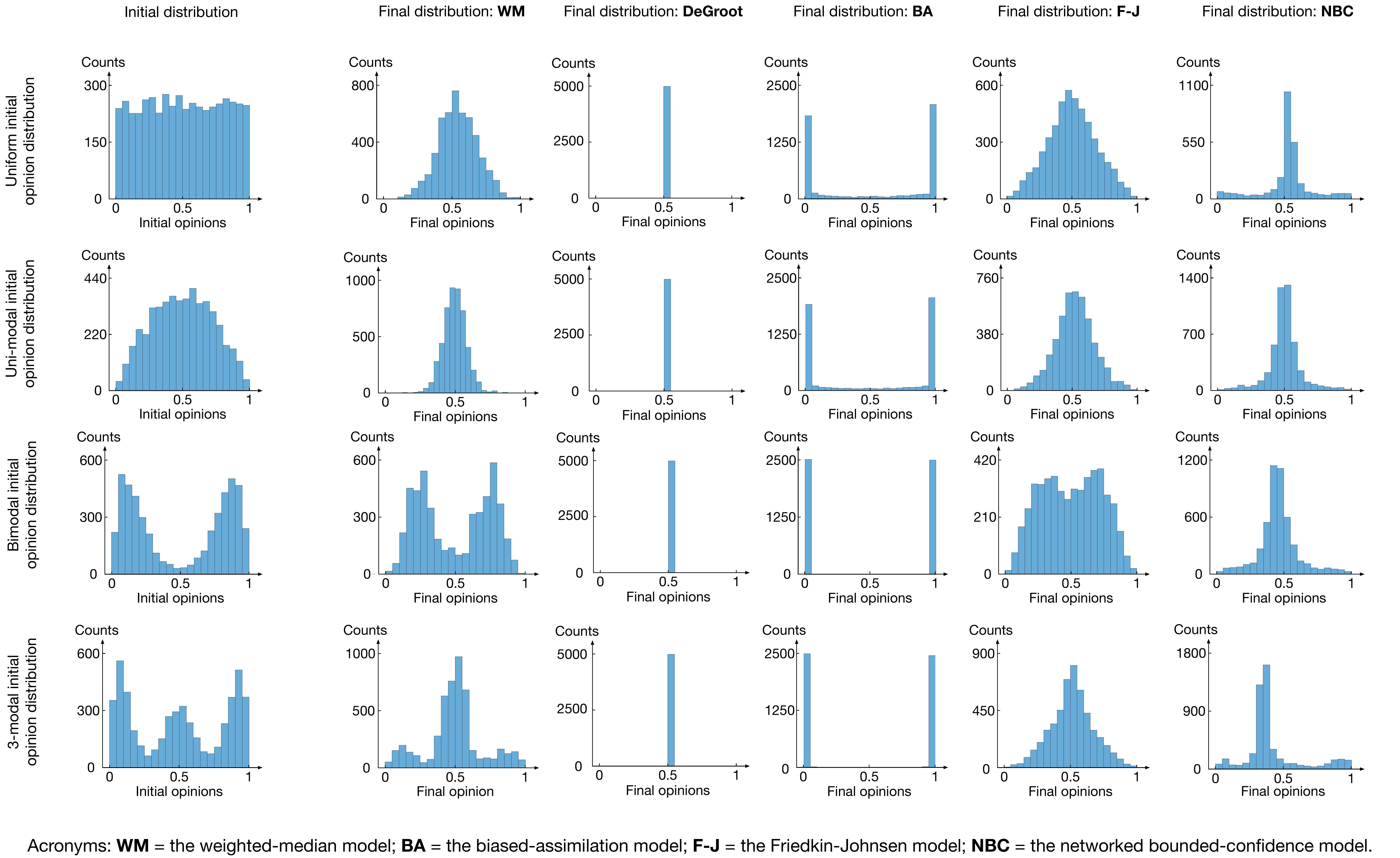}
\caption{Distributions of the initial opinions and the final opinions predicted by different models. The
simulations are run on the same small-world network with
5000 nodes.}\label{fig:initial-final-WS}
\end{figure*}

\end{document}